\documentclass[a4paper]{article}
\usepackage{amsfonts,amssymb,amsmath,amsthm,cite}
\usepackage{ucs} 
\usepackage[utf8x]{inputenc}
\usepackage{graphicx}
\usepackage[english]{babel}
\usepackage{slashed}
\usepackage{textcomp}
\usepackage{tikz}
\usepackage[compat=1.1.0]{tikz-feynman}
\usetikzlibrary{shapes,arrows,calc}
\catcode`@=12

\evensidemargin=\oddsidemargin
\newcommand{\pp}[2]{\frac{\partial #1}{\partial #2}}

\newcommand{\ff}[2]{\frac{#1}{#2}}
\newcommand{\be}[0]{\beta }
\newcommand{\al}[0]{\alpha }

\newcommand{\hh}[2]{\frac{d #1}{d #2}}

\newcommand{\ma}[1]{\mathrm{#1}}

\newcommand{\ppp}[2]{\frac{\partial^2 #1}{\partial #2^2}}

\newcommand{\mb}[1]{\mathbb{#1}}
\newcommand{\pur}[1]{\partial_{#1}}
\newcommand{\lr}[1]{\left( {#1} \right)}
\newcommand{\tib}[0]{\tilde{\beta}}

\newtheorem{assumption}{Assumption}[section]
\newtheorem{theorem}{Theorem}

\newtheorem{corollary}[assumption]{Corollary}

\newtheorem{lem}{Exercise}
\newtheorem{lemma}[assumption]{Lemma}

\newtheorem{definition}{Definition}

\newcommand{\fel}[3]{\feynmandiagram [baseline=-0.5ex, small, horizontal=a to b] {
		a [particle=$#1$]-- [ edge label=$#2$] b [particle=$#3$],
	};}

\newcommand{\fed}[4]{\feynmandiagram [baseline=-0.5ex, small, layered layout, horizontal=a to b] {
		a [particle=$#1$]--[ edge label=$#2$] c[dot]-- [ edge label=$#3$] b [particle=$#4$]  ],
	};}

\newcommand{\mlg}[3]{
	\begin{tikzpicture}[baseline= #1 ex]
	\def\sc{#2};
	\def\ss{#3};
	\begin{feynman}
	\vertex (a1);
	\vertex [blob,style={pattern= north west lines}, scale=\sc* \ss, right=-0.00*\ss*\sc cm of a1] (t2) {};
	\vertex [blob,style={pattern= north east lines}, scale=\sc* \ss, right=-0.00*\ss*\sc cm of a1] (t1) {};
	\diagram* {
		(a1) };
	\end{feynman}
	\end{tikzpicture}
}

 \newcommand{\gol}[4]{
	\begin{tikzpicture}[baseline= #1 ex]
	\def\sc{#2};
	\def\ss{#3};
	\begin{feynman}
	\vertex (a1);
	\vertex [right=\sc*0.5 cm of a1] (a2) ;
	\vertex [blob,style={pattern= north west lines}, scale=\sc* \ss, right=-0.004*\ss*\sc cm of a2] (t2) {};
	\vertex [blob,style={pattern= north east lines}, scale=\sc* \ss, right=-0.004*\ss*\sc cm of a2] (t1) {};
	\diagram* {
		(a1) -- [insertion={[size=#4 cm]0}] (a2)};
	\end{feynman}
	\end{tikzpicture}
}

\newcommand{\chups}[4]{
	\begin{tikzpicture}[baseline= #1 ex]
	\def\sc{#2};
	\def\ss{#3};
	\begin{feynman}
	\vertex (a1);
	\vertex [right=\sc cm of a1] (a2);
	\vertex [right=\sc cm of a2] (a3);
	\vertex at ($(a2)!0.5!(a3)! \sc*0.5 cm !90:(a2)$) (z1);
	\vertex at ($(a2)!0.5!(a3)!-\sc*0.5 cm!90:(a2)$) (z2);
	\diagram* {
		(a2) -- [quarter right] (z1) -- [quarter right] (a3),
		(a2) -- [quarter left] (z2) -- [quarter left] (a3),
		a1 -- [insertion={[size=#4 cm]0}] a2};
	\end{feynman}
	\end{tikzpicture}
}

\newcommand{\och}[3]{
	\begin{tikzpicture}[baseline= #1 ex]
	\def\sc{#2};
	\def\ss{#3};
	\begin{feynman}
	\vertex (a1);
	\vertex [right=\sc cm of a1] (a2) ;
	\vertex at ($(a1)!0.5!(a2)! \sc*0.5 cm !90:(a2)$) (z1);
	\vertex at ($(a1)!0.5!(a2)!-\sc*0.5 cm!90:(a2)$) (z2);
	\vertex[right=\sc cm of a2] (a3);
	\vertex[right= \sc cm of a3] (a4);
	\vertex at ($(a3)!0.5!(a4)!\sc*0.5 cm!90:(a4)$) (z3);
	\vertex at ($(a3)!0.5!(a4)!-\sc*0.5 cm!90:(a4)$) (z4);
	\vertex [dot, scale=\sc* \ss, right=\sc cm-0.133*\ss*\sc cm of a1] (t2) {a};
	\vertex [dot, scale=\sc* \ss, right=\sc cm-0.133*\ss*\sc cm of a2] (t3) {a};
	\diagram* {
		(a1) -- [quarter left] (z1) -- [quarter left] (a2),
		(a1) -- [quarter right] (z2) -- [quarter right] (a2),
		(a3) -- [quarter left] (z3) -- [quarter left] (a4),
		(a3) -- [quarter right] (z4) -- [quarter right] (a4),
		a2--a3};
	\end{feynman}
	\end{tikzpicture}
}

\newcommand{\lemm}[3]{
	\begin{tikzpicture}[baseline= #1 ex]
	\def\sc{ #2 };
	\def \ss{#3 };
	\begin{feynman}
	\vertex (a1);
	\vertex[right=\sc cm of a1] (a2);
	\vertex [dot, scale=\sc* \ss, right=\sc cm-0.133*\ss*\sc cm of a1] (t2) {a};
	\vertex [dot, scale=\sc* \ss, left=\sc cm-0.133*\ss*\sc cm of a2] (t3) {a};
	\diagram*{
		(a1) -- [half left, in=90, out=90] (a2),(a1) -- [half right, in=-90, out=-90] (a2), (a1)--(a2)
	};
	\end{feynman}
	\end{tikzpicture}
}

\newcommand{\pov}[3]{
	\begin{tikzpicture}[baseline= #1 ex]
	\def\sc{#2};
	\def\ss{#3};
	\begin{feynman}
	\vertex (a1);
	\vertex [right=\sc*0.5 cm of a1] (a2) ;
	\vertex [right=\ss*0.75*\sc cm of a2] (a3) ;
	\vertex [right=\sc*0.5 cm of a3] (a4) ;
	\vertex [blob, style={pattern= north west lines}, scale=\sc* \ss, right=-0.004*\ss*\sc cm of a2] (t2) {};
	\vertex [blob, style={pattern= north east lines}, scale=\sc* \ss, right=-0.004*\ss*\sc cm of a2] (t1) {};
	\diagram* {
		(a1) -- (a2), (a3) -- (a4)};
	\end{feynman}
	\end{tikzpicture}
}

 \newcommand{\mg}[3]{
	\begin{tikzpicture}[baseline= #1 ex]
	\def\sc{#2};
	\def\ss{#3};
	\begin{feynman}
	\vertex (a1);
	\vertex [blob,style={pattern= north west lines}, scale=\sc* \ss, right=-0.00*\ss*\sc cm of a1] (t2) {};
	\diagram* {
		(a1) };
	\end{feynman}
	\end{tikzpicture}
}

\newcommand{\gl}[4]{
	\begin{tikzpicture}[baseline= #1 ex]
	\def\sc{#2};
	\def\ss{#3};
	\begin{feynman}
	\vertex (a1);
	\vertex [right=\sc*0.5 cm of a1] (a2) ;
	\vertex [blob,style={pattern= north west lines}, scale=\sc* \ss, right=-0.004*\ss*\sc cm of a2] (t2) {};
	\diagram* {
		(a1) -- [insertion={[size=#4 cm]0}] (a2)};
	\end{feynman}
	\end{tikzpicture}
}

\begin{document}
\begin{center}
{\Large\textbf{Some Basic Tools of QFT}}
\vspace{0.5cm}
		
{\large A.~V.~Ivanov$^\dag$ and M.~A.~Russkikh$^\ddag$}
		
\vspace{0.5cm}
		
{\it St. Petersburg Department of Steklov Mathematical Institute of
Russian Academy of Sciences,}\\{\it 27 Fontanka, St. Petersburg 191023, Russia}\\
{\it Leonhard Euler International Mathematical Institute, 10 Pesochnaya nab.,}\\
{\it St. Petersburg 197022, Russia}\\
$^\dag${\it E-mail: regul1@mail.ru} ~~~~~~
$^\ddag${\it E-mail: russkix-maksim@mail.ru}	
\end{center}
	
\date{\vskip 20mm}
\vskip 10mm
\begin{abstract}
This work has a methodological nature and is a set of lecture notes for undergraduate students. It is devoted to the study of the basic tools of quantum field theory on the example of the simplest cubic "toy" model. We introduce such concepts as the functional integral, the generating functions, the background field method and the Feynman diagram technique, and also consider relations among them. The model under consideration allows us to perform all the calculations explicitly, which significantly increases the visibility and clarity of the notes. We also provide neat proofs that are usually omitted in standard educational publications.
\end{abstract}

\tableofcontents
\newpage
\section{Introduction}
Quantum field theory plays one of the key roles in modern theoretical physics, see \cite{bogoliubov-1980,peskin-1995,wienberg-1995,zee-2003,schwartz-2014}. It has a long history of development and includes many auxiliary methods, such as a functional integration \cite{daniell-1919,takhtajan-2008,ivanov-2019}, the Feynman diagram technique \cite{feynman-1948,feynman-2005}, generating functions \cite{lando-2003,zinn-2005}, the background field method \cite{abbott-1982,arefeva-1974,faddeev-2009,ivanov-kharuk-2019}, and many others \cite{vasiliev-1998}. These methods are useful and popular tools in modern science.

At the moment, there are a large number of scientific and educational publications devoted to this topic. However, as a rule, such textbooks, partially mentioned above, contain a very long and thorough introduction, and most of the methods are formulated in a general form. This way reduces the visibility and, in our opinion, does not give the reader a sufficient understanding of the subtleties. At the same time, the choice of complex models as an example does not allow to perform all the calculations explicitly.

In our work, we propose a sequential description of some basic tools of quantum field theory using a cubic "toy" model as an example. This approach allows us to make calculations explicitly, while leaving the necessary key laws and properties intact. It should be noted that such the model has already appeared in the methodological notes earlier \cite{cvitanovic-1983}. The difference of our attempt is that we consider the model in more detail and prove all the necessary diagram identities.

The work has the following structure. First, in Section \ref{pro}, we introduce a definition of a special integral on an infinite-dimensional space and consider its basic properties. In particular, we present a standard calculation method. Then we introduce the cubic model and discuss its properties, such as the perturbative decomposition and the connection with Airy functions. We also pay attention to the possibility of factorization, that leads to significant simplifications.

In Section \ref{diag}, we explain the diagram technique. At the same time, we consider it from a mathematical and physical point of view. As it is known, such approaches have some differences. We also prove several important theorems: on multiple exponential differentiation and on the relation between connected and strongly connected diagrams.

Then, in Section \ref{gen}, we give definitions of a set of generating functions, derive the equations that they satisfy, including the quantum equation of motion, and find the solutions. In Section \ref{res}, we give a diagrammatic interpretation for the equations and the generating functions, and also prove them. 

Further, in Sections \ref{back} and \ref{conc}, we explain the background field method by direct calculations and give some final comments. In particular, we present a method of transition from the simplest cubic model to an arbitrary one.

\section{Problem statement}\label{pro}
\subsection{Functional integral}
Let us introduce an infinite-dimensional space $V=\mathbb{R}^{\infty}$, that is the infinite product of copies of $\mathbb{R}$. An element $x\in V$ is the sequence $\{x_k\}_{k\geqslant1}$. Next, for two elements $x, y\in V$ we define a scalar product by the formula $(x,y)=\sum_{k\geqslant1}x_ky_k$, as well as two special maps: $V\times V\to V$, that operates according to the rule $xy=\{ x_k y_k\}_{k\geqslant1}$, and $V\times \mathbb{R}\to V$, that is defined by $x^c=\{ x_k^c\}_{k\geqslant1}$ for all $c\in\mathbb{R}$.
Then, for all $x \in V$ we define $\mathbb{X}_N\in V$, where $N \in \mathbb{N}$, as follows: $(\mathbb{X}_N)_k=x_k$, if $k \leqslant N $, and $(\mathbb{X}_N) _k=0$, if $k> N$. We also introduce the set of positive sequences
$\mathcal{B}=\{\beta\in V|\,\beta_k>0, k\in \mathbb{N}\}$.

\begin{definition}
\label{def11}
Let $ \beta\in\mathcal{B}$, and $F:V\to\mathbb{C}$ be a functional on $V$, then a functional integral of $F$ over $V$ equals to the number
	\begin{equation}
	\label{defint}
	\Phi_\beta(F\,)=\lim_{N\rightarrow \infty} 
	\Bigg(\prod\limits_{k=1}^N\frac{\beta_k}{2\pi}\Bigg)^{1/2}
	\int_{\mathbb{R}^N}F\left( \mathbb{X}_N\right)e^{-\frac{1}{2}\left(\beta\mathbb{X}_N,\mathbb{X}_N\right)}\,d\mathbb{X}_N,
	\end{equation}
	where $d\mathbb{X}_N=dx_1...dx_N$.
\end{definition}

As a rule, the latter construction has properties of the standard integral. Indeed, let $F$ and $G$ be two functionals on $V$, and $a, b \in\mathbb{C}$, then we can write out the linearity, conjugation, and monotonicity:
\begin{equation*}
\Phi_\beta(aF+bG\,)=a\Phi_\beta(F\,)+b\Phi_\beta(G\,),\,\,\,\,
\Phi_\beta(\overline{F\,})=\overline{\Phi_\beta(F\,)}\,,
\end{equation*}
\begin{equation*}
F(x)\geqslant G(x)\,\,\mbox{for all}\,\,x\in V\Rightarrow \Phi_\beta(F\,)\geqslant \Phi_\beta(G\,).
\end{equation*}

Proof of such properties for a particular class of functionals is a separate task and does not fall within the scope of this paper.

As examples, we give two types of functionals that often occur in calculations
\begin{equation}
\Phi_\beta\left(1\right)
=\lim_{N\rightarrow \infty} \prod\limits_{k=1}^N\left(\sqrt{ \frac{\beta_k}{2\pi}}\int_{\mathbb{R}}e^{-\frac{1}{2}\beta_kx_k^2}\,dx_k\right)=\lim_{N\rightarrow \infty} \prod\limits_{k=1}^N1=1,
\end{equation}
\begin{equation}
\label{examp}
\Phi_\beta\left(e^{(x,y)}\right)
=\lim_{N\rightarrow \infty} \prod\limits_{k=1}^N\left(\sqrt{ \frac{\beta_k}{2\pi}}\int_{\mathbb{R}}e^{x_k y_k}e^{-\frac{1}{2}\beta_k x_k^2}\,dx_k\right)=\lim_{N\rightarrow \infty} \prod\limits_{k=1}^N e^{\frac{y^2_k}{2\beta_k}}=e^{\frac{1}{2}(\tib y,y)},
\end{equation}
where $\tib_k=1/\beta_k$.
\begin{lem}
Find some additional examples of functionals, that can be integrated explicitly.
\end{lem}

\subsection{Method of calculation}
Now let us consider one way to work with the functional integral. The fact is that in most cases it is impossible to calculate the integral explicitly, so you have to use  an asymptotic expansion in a small parameter.
\begin{lemma}
\label{lem1}
Let $x,y\in V$, and $F(x)$ be a polynomial of a finite degree, then we have
\begin{align}
F(x)=F(\partial_y)e^{(y,x)}\Big|_{y=0}.
\end{align}
\end{lemma}
\begin{proof}
It is enough to use the exponential differentiation and the equality $[\partial_{x^i},x_j]=\delta_{ij}$.
\end{proof}
\begin{lemma}
\label{lem2}
Let $y\in V$, and $F$ be a functional on $V$, then we have
\begin{equation}
\label{phif}
\Phi_{\beta}(F\,)=F(\partial_{y})e^{\frac{1}{2}(\tib y,y)} \Big|_{y=0}.
\end{equation}
\end{lemma}
\begin{proof} 
Let us use Definition \ref{def11}, Lemma \ref{lem1}, and the result from example (\ref{examp}). Then we have the following chain of equalities
\begin{align*}
\Phi_{\be}(F\,)&
=\lim_{N\rightarrow \infty} \Bigg(\prod\limits_{k=1}^N \frac{\be_k}{2\pi}\Bigg)^{\frac{1}{2}}
\int_{\mathbb{R}^N}F(\partial_{\mathbb{Y}_N})
e^{(\mathbb{Y}_N,\mathbb{X}_N)}\Big|_{\mathbb{Y}_N=0}
e^{-\frac{1}{2}(\beta\mathbb{X}_N,\mathbb{X}_N)}\,d\mathbb{X}_N\\
&=\lim_{N\rightarrow \infty}F(\pur{\mb{Y}_N}) \prod\limits_{k=1}^N e^{\ff{y_k^2}{2 \be_k}}\Big|_{y_k=0}=F(\partial_{y})e^{\frac{1}{2}(\tib y,y)}\Big|_{y=0},
\end{align*}
from which the statement follows.
\end{proof}
It is worth noting that in physical applications the functional $F$ usually contains a small parameter. In view of this, we can use the Taylor expansion in powers of this small parameter. Moreover, the coefficients of the expansion are commonly polynomials of a finite degree. Thus, applying Lemma \ref{lem2} to each coefficient, we obtain the asymptotic series over the small parameter.
\begin{lem}
Apply Lemma \ref{lem2} to the functional $\exp(\alpha,y)$, where $\alpha\in\mathcal{B}$. Calculate it explicitly and find such $\alpha$, for which the answer is finite.
\end{lem}
\begin{lem}
\label{pfv}
Show that the formula $\Phi_{\beta}(F(\,\cdot\,))=\Phi_{\mathbf{1}}(F(\beta^{-\frac{1}{2}}\,\cdot\,))$ holds, where $\mathbf{1}=\{1\}_{k\geqslant1}\in\mathcal{B}$.
\end{lem}
\begin{corollary}
\label{cor}
From Lemma \ref{lem2} and Exercise \ref{pfv} it follows that
\begin{equation}
F(\partial_{y})e^{\frac{1}{2}(\tib y,y)} \big|_{y=0}=
F(\beta^{-\frac{1}{2}}\partial_{y})e^{\frac{1}{2}(y,y)} \big|_{y=0}.
\end{equation}
\end{corollary}

\subsection{Cubic functional}
Let us move on to a more meaningful example and consider the following type of functionals: $F_\sigma(x)=\exp\big\{i\sum_{k\geqslant1}x_k^3\sigma_k^{\phantom{3}}\big\}$, where $\sigma \in V$.
In this case, explicit factorization occurs, and using Definition \ref{def11}, we can write the functional integral as
\begin{equation}
\label{cub}
\Phi_\beta(F_{\sigma})=
\lim_{N\rightarrow \infty} \prod\limits_{k=1}^N\Bigg(\sqrt{ \frac{\beta_k}{2\pi}}
\int_{\mathbb{R}}e^{i x_k^3 \sigma_k}e^{-\frac{1}{2}x_k^2\beta_k}\,dx_k\Bigg).
\end{equation}

It is convenient to make the change $x_k\to x_k\sqrt {\beta_k}$ for all $k\in\mathbb{N}$, and then move from $\sigma$ to $\alpha=\{\sigma_k^{\phantom{3}}/\beta^{3/2}_k\}_{k\geqslant1}$. Then the formula (\ref{cub}) can be rewritten as
\begin{equation}
\label{intI}
\Phi_\beta(F_{\sigma})=
\Phi_{\mathbf{1}}(F_{\alpha})=
\prod_{k\geqslant 1}I(\alpha_k),\,\,\,\,\mbox{where}\,\,\,\,
I(\alpha_k)=\frac{1}{\sqrt{2\pi}}\int_{\mathbb{R}}e^{-\frac{1}{2}x^2+i\alpha_k x^3}\,dx.
\end{equation}

Let us note the obvious properties of the obtained functions:
\begin{itemize}
	\item $I(0)=1$;
	\item $|\exp(i\alpha_kx^3)|\leqslant1$, therefore $|I(\alpha_k)|\leqslant1$ and $|\Phi_\beta(F_{\sigma})|\leqslant1$;
	\item $I(\alpha_k)=I(-\alpha_k)$;
	\item $I(\alpha_k)=\mathrm{Re}{(I(\alpha_k))}>0$;
	\item $I(\alpha_k)$ decreases monotonously on $\mathbb{R}_+$ from $1$ to $0$.
\end{itemize}
\begin{lemma}
\label{lem3}
Let $I (\alpha_k)$ be the function defined above, then it satisfies the differential equation
\begin{equation}
\left[27 \alpha_k^3\ppp{}{\alpha_k}+(1+81\alpha_k^2)\pp{}{\alpha_k}+15\alpha_k\right]I(\alpha_k)=0.
\end{equation}
\end{lemma}
\begin{proof}
First of all we introduce an auxiliary function of the form
\begin{equation*}
I_n=\int_{\mathbb{R}}x^ne^{-\frac{x^2}{2}+i\alpha_k x^3}dx.
\end{equation*}

Using integration by parts, we make sure that $I$ satisfies the relation
$I_n=(n-1)I_{n-2}+3 i \alpha_k I_{n+1}$.
Let us apply this formula twice to $I_3$
\begin{equation*}
I_3=2I_1+3i\alpha_k^{\phantom{2}} I_4=15i\alpha_k^{\phantom{2}} I_2-9\alpha_k^2 I_5=
15 i \alpha_k^{\phantom{2}} I_0-81 \alpha_k^2 I_3-i 27 \alpha_k^3 I_6,
\end{equation*}
and then use the relations $I_0=I$, $I_3=-i\partial_{\alpha_k}^{\phantom{2}}I$, and $I_6=-\partial^2_{\alpha_k}I$, from which the statement of the lemma follows.
\end{proof}

\begin{lemma}
The function $I(\alpha_k)$ has the following equivalent representation
\begin{equation}
\label{Ial}
I(\alpha_k)=\frac{\sqrt{2\pi}}{(3\alpha_k)^{1/3}} e^{\frac{1}{108 \alpha_k^2}} \ma{Ai}\Bigg(\frac{1}{4\lr{3\alpha_k}^{4/3}}\Bigg),
\end{equation}
where $\ma{Ai}$ is the Airy function \cite{vallee-2010}. $I(\alpha_k)$ is drawn in Figure \ref{pic1}.
\end{lemma}
\begin{proof}
The result follows from Lemma \ref{lem3}, as well as two boundary conditions that are obtained by explicit differentiation of (\ref{intI}). A clear derivation in a more general case we give in Section \ref{gen}.
\end{proof}
\begin{figure}[h]
	\center{\includegraphics[width=0.3\linewidth]{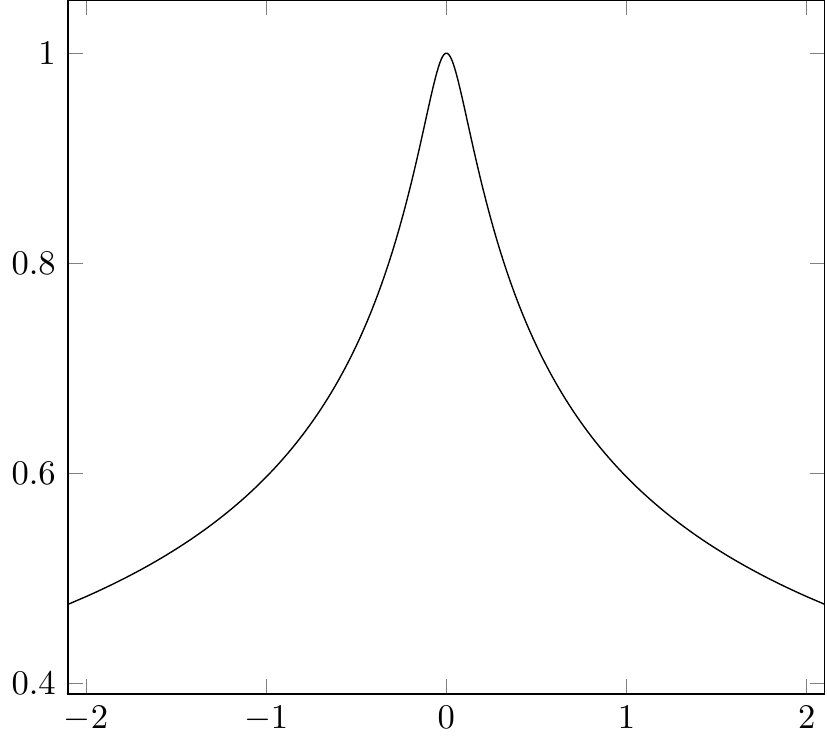}}
	\caption{The function $I(\alpha_k)$ for $\alpha_k\in[-2.1,2.1]$.}
	\label{pic1}
\end{figure}

\begin{lem}
Verify formula (\ref{Ial}) by using the equation from Lemma \ref{lem3}.
\end{lem}

\subsection{Asymptotic expansion}

\begin{lemma}
\label{lem9}
If $\Phi_{\mathbf{1}}(F_{\alpha})>0$, then $\alpha_k\to0$.
\end{lemma}
\begin{proof}
We rewrite the functional integral as follows
\begin{equation*}
\Phi_{\mathbf{1}}(F_{\alpha})=\prod_{k\geqslant 1}I(\alpha_k)
=\ma{\exp}\Bigg(\sum_{k\geqslant 1} \ma{ln}(I(\alpha_k))\Bigg).
\end{equation*}

We need the sum in the exponential is finite. Hence, given the property $I (\alpha_k)\in (0,1]$, 
the relation $\lim_{k\to\infty}I(\alpha_k)=1$ must be satisfied. Taking into account the properties described above, it is equivalent to $\lim_{k\to \infty}\alpha_k=0$, from which the statement follows.
\end{proof}
\begin{lem}
Find how fast the sequence $\alpha_k$ should tend to zero in Lemma \ref{lem9}.
\end{lem}

\begin{lemma}
\label{lem10}
The asymptotic expansion takes place
\begin{align}
\Phi_{\mathbf{1}}(F_{\alpha})=\prod_{k\geqslant 1} \sum_{n\geqslant0} \frac{(-1)^{n}}{8^{n}}\frac{(6n)!}{(2n)!(3n)!}\alpha^{2n}_k.
\end{align}
\end{lemma}
\begin{proof}
Let us use Lemma \ref{lem2} to calculate $\Phi_{\mathbf{1}}(F_{\alpha})$. In view of the factorization of the problem (\ref{intI}), we apply Lemma \ref{lem2} in the one-dimensional case
\begin{equation*}
e^{i \alpha \partial_{y}^{3}} e^{\frac{1}{2}y^2}\bigg|_{y=0}= \sum_{k\geqslant 0}\sum_{m\geqslant 0}\frac{(i \alpha)^k }{k!}\frac{1}{2^m}\frac{1}{m!} \partial_{y}^{3k} y^{2m}\bigg|_{y=0}=
\sum_{n\geqslant 0}\frac{(-1)^{n}}{8^{n}}\frac{(6n)!}{(2n)!(3n)!}\alpha^{2n},
\end{equation*}
where in the second equality we have used $3k=2m=6n$.
\end{proof}

\section{Diagrams}\label{diag}
\subsection{Mathematical point of view}
As it was noted above, see formula (\ref{phif}), functional integral can be evoluated by using the differentiating of the exponential $\exp\big[(\beta y,y)/2\big]$, where $y\in V$ and $\beta\in\mathcal{B}$. Let us draw the attention on the quadratic form $(\beta y,y)$, the kernel of which is $\beta_i\delta_{ij}$. Now we are going to formulate diagram technique rules, that give a way to visualize the calculations.
\begin{definition}
\label{def}
The kernel $\beta_i\delta_{ij}$ is denoted by a line with three indices 
$\feynmandiagram [baseline=-0.5ex, small, horizontal=a to b]
{a [particle=$i$]-- [ edge label=$\beta$] b [particle=$j$]],};$.
By black dot 
$\,\feynmandiagram [baseline=-0.7ex, small, layered layout, horizontal=a to b] {
	c[dot] ],};$
we define a summation over the set $\mathbb{N}$.
The summation of one side of the kernel with an element from $x\in V$ is denoted by replacement of the integer index by $x$. If $\beta=\mathbf{1}$, then the third index may be omitted.
\end{definition}
Let us give some examples of using the diagram technique rules: if $\alpha,\beta\in\mathcal{B}$ and $x\in V$, then we have

\begin{equation*}
\sum_{j\geqslant 1}\beta_i \delta_{ij}x_j
=\feynmandiagram [baseline=-0.5ex, small, horizontal=a to b] {
	a [particle=$i$]-- [ edge label=$\be$] b [particle=$x$]  ],};,\,\,\,\,\,\,\,\,\,\,
(\beta x,x)=\sum_{i\geqslant 1}x_i\beta_ix_i=\sum_{i,j\geqslant 1}x_i\beta_i \delta_{ij} x_j
=\feynmandiagram [baseline=-0.5ex, small, horizontal=a to b]
{a [particle=$x$]-- [ edge label=$\be$] b [particle=$x$]  ],};,
\end{equation*}
\begin{equation*}
\sum_{k\geqslant 1}\fel{i}{\alpha}{k} \fel{k}{\beta}{j}=\sum_{k\geqslant 1} \alpha_i\delta_{ik} \beta_k\delta_{kj}=
\feynmandiagram [baseline=-0.5ex, small, layered layout, horizontal=a to b] {
	a [particle=$i$]--[ edge label=$\alpha$] c[dot]-- [ edge label=$\beta$] b [particle=$j$]  ],};.
\end{equation*}

Such diagrams have some remarkable properties, that play an important role in calculations. Among them, we note the symmetry, chain rule, and transformation after a differentiation:
\begin{equation}
\label{prop1}
\fel{i}{\beta}{x}=\fel{x}{\beta}{i},\,\,\,\,\,\,\,\,\,\,
\fed{i}{\al}{\be}{j}=\fel{i}{\alpha\beta}{j},
\end{equation}
\begin{equation}
\label{ddx}
\hh{}{x_i} \fel{x}{\beta}{x}=2 \fel{i}{\beta}{x},\,\,\,\,\,\,\,\,\,\,
\hh{}{x_j}\fel{i}{\beta}{x}=\fel{i}{\beta}{j}.
\end{equation}

The latter relations can be obtained by using Definition \ref{def} and examples given above. A natural question arises: why do we introduce a diagram technique in the case of the $\delta$-kernel? The fact is that in most physical models such relations appear, but they have more complicated kernels. Therefore, the diagrams introduced above are general in nature.

Another important property is called star-triangle relation and has the following view: if $\alpha,\beta,\gamma\in\mathcal{B}$ and $i,j,k\in\mathbb{N}$, then we have the equality depicted in Figure \ref{star}.

\begin{figure}[h]
\center{\includegraphics{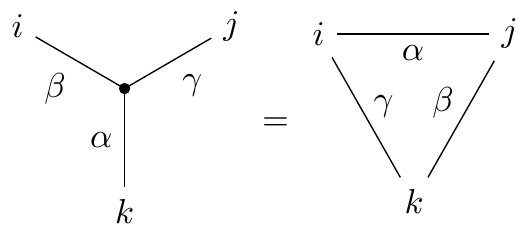}}
\caption{Star-triangle relation.}
\label{star}
\end{figure}

The latter relation is proved immediately by substituting the explicit form for kernels. Indeed, on both sides we have a nonzero result only when $i=j=k$. In the case both sides are equal to $\alpha_i\beta_i\gamma_i$. It is also worth noting that we can generalize the result to any number of vertices. However, the generalization appears due to the simplicity of the example.

\subsection{Exponential decomposition}
Now we are going to understand how to calculate the $n$-th derivative for $\exp\big[\ff{1}{2} \fel{x}{}{x}\big]$, where we have used Definition \ref{def} and Corollary \ref{cor}, according to which we can analyze only the case $\beta=\mathbf{1}$.
First of all, for illustration let us give the expressions for the first and second derivatives: if $i,j\in\mathbb{N}$, then
\begin{equation*}
\frac{d}{dx_i}e^{\frac{1}{2}\fel{x}{}{x}}=\big(\fel{i}{}{x}\big) e^{\frac{1}{2}\fel{x}{}{x}},
\end{equation*}
\begin{equation*}
\frac{d^2}{dx_jdx_i}e^{\frac{1}{2} \fel{x}{}{x}}=\bigg(\fel{i}{}{j}+
\begin{smallmatrix}\fel{i}{}{x}\\\fel{j}{}{x}
\end{smallmatrix}
\bigg)e^{\frac{1}{2}\fel{x}{}{x}}.
\end{equation*}

\begin{lem}
Find expressions for the third and the fourth derivatives.
\end{lem}

Here we make some comments at this stage. After each differentiation, the parity of the construction with respect to the variable $x$ changes. Also, after differentiation, the exponential has a factor, that is constructed from the elements of two types: $\fel{i}{}{j}$ and $\fel{i}{}{x}$, where $i,j\in\mathbb{N}$.

Then, after applying the differentiation operator by the variables $x_{i_1},\ldots,x_{i_k}$, where $k\in\mathbb{N}$, each term is constructed by using diagrams, described above. Hence, their ends contain the points $i_1,\ldots,i_k\in\mathbb{N}$, and may also contain $x$-variable $m$ times, where $m \in \{j:\,0\leqslant j \leqslant k,\,\mbox{and}\,k-j\in 2 \mathbb{N}\cup\{0\} \}$.

\begin{definition}
\label{def1}
Let $k,n\in\mathbb{N}$, such that $k\geqslant n$. Let also $i_1,\ldots,i_k\in\mathbb{N}$ and $y_1,\ldots,y_n\in V$. So we have two sets of indices $\begin{smallmatrix}i_1,\ldots,i_k\\y_1,\ldots,y_n\end{smallmatrix}$.
By round brackets 
$\big(\begin{smallmatrix}i_1,\ldots,i_k\\y_1,\ldots,y_n\end{smallmatrix}\big)$ 
we define an operation of joining indices from different rows by the lines in all possible ways. Then, if $k-n$ is even, square brackets
$\big[\big(\begin{smallmatrix}i_1,\ldots,i_k\\y_1,\ldots,y_n\end{smallmatrix}\big)\big]$ denote the connection of the remaining indices in the first row by the lines in all possible ways. If $k-n$ is odd, then $\big[\big(\begin{smallmatrix}i_1,\ldots,i_k\\y_1,\ldots,y_n\end{smallmatrix}\big)\big]=0$.

Hence, we can define a function by the following equality
\begin{equation}
\label{hfun}
h_{k,n}(i_1,\ldots,i_k;y_1,\ldots,y_n)=\frac{1}{n!}\big[\big(\begin{smallmatrix}i_1,\ldots,i_k\\y_1,\ldots,y_n\end{smallmatrix}\big)\big].
\end{equation}
\end{definition}

Let us give some examples:
\begin{equation*}
\big[\big(\begin{smallmatrix}i,j\\y_1,y_2\end{smallmatrix}\big)\big]\big|_{y_1=y_2=x}=
\begin{smallmatrix}\fel{i}{}{x}\\\fel{i}{}{x}\end{smallmatrix},
\end{equation*}
\begin{equation*}
\big[\big(\begin{smallmatrix}i,j,k\\x\end{smallmatrix}\big)\big]=
\begin{smallmatrix}\fel{i}{}{j}\\\fel{k}{}{x}\end{smallmatrix}+
\begin{smallmatrix}\fel{i}{}{k}\\\fel{j}{}{x}\end{smallmatrix}+
\begin{smallmatrix}\fel{j}{}{k}\\\fel{i}{}{x}\end{smallmatrix}.
\end{equation*}

\begin{theorem}
\label{th}
Let $k,i_1,\ldots,i_k\in\mathbb{N}$, $x=\{x_k\}_{k\geqslant 1}\in V$. Then, using Definition \ref{def1}, we have the following equality
\begin{equation}
\label{form}
\frac{d^k}{dx_{i_k}\ldots dx_{i_1}}\,e^{\frac{1}{2} \fel{x}{}{x}}=g_k(i_1,\ldots,i_k;x)\,e^{\frac{1}{2} \fel{x}{}{x}},
\end{equation}
where
\begin{equation}
\label{form2}
g_k(i_1,\ldots,i_k;x)=\sum_{n=0}^{k}h_{k,n}(i_1,\ldots,i_k;y_1,\ldots,y_n)\big|_{y_1=\ldots=y_n=x}.
\end{equation}
\end{theorem}
\begin{proof} It follows from equality (\ref{form}), that the relations should be hold
\begin{equation*}
g_{p+1}(i_1,\ldots,i_{p+1};x)=
\frac{d}{dx_{p+1}}g_p(i_1,\ldots,i_p;x)+
g_p(i_1,\ldots,i_p;x)(\fel{i_{p+1}}{}{x}).
\end{equation*}

Let us check, that the series (\ref{form2}) satisfies the last relation. First of all we derive some additional properties for $h_{k,n}$-functions. Let us transfer an element from under the round brackets
\begin{align*}
h_{p+1,n}(i_1,\ldots,i_{p+1};y_1,\ldots,y_n)&=
\frac{1}{n!}
\big[\big(
\begin{smallmatrix}i_1,\ldots,i_{p+1}\\y_1,\ldots,y_n\end{smallmatrix}\big)
\big]\\&=
\frac{1}{n!}
\big[\big(
\begin{smallmatrix}i_1,\ldots,i_p\\y_1,\ldots,y_n\end{smallmatrix}\big)
\begin{smallmatrix}i_{p+1}\\\phantom{x}\end{smallmatrix}\big]+
\sum_{j=0}^{n}\frac{1}{n!}
\big[\big(
\begin{smallmatrix}i_1,\ldots,i_{p}\\y_1,\ldots,\hat{y}_j,\ldots,y_n\end{smallmatrix}\big)
\big](\fel{i_{p+1}}{}{y_j}),
\end{align*}
where the hat $\hat{y}_j$ indicates that the element is missed. Then it is easy to verify that the relation holds 
\begin{equation*}
\frac{1}{n!}
\big[\big(
\begin{smallmatrix}i_1,\ldots,i_p\\y_1,\ldots,y_n\end{smallmatrix}\big)
\begin{smallmatrix}i_{p+1}\\\phantom{x}\end{smallmatrix}\big]=
\frac{n+1}{(n+1)!}
\big[\big(
\begin{smallmatrix}i_1,\ldots,i_p\\y_1,\ldots,y_n,i_{p+1}\end{smallmatrix}\big)
\big].
\end{equation*}

After that, applying a symmetry of the
$h_{p+1,n}(i_1,\ldots,i_{p+1};y_1,\ldots,y_n)$ under the permutations of the group of the indices $y_1,\ldots,y_n$, and choosing the last points in the form $y_1=\ldots=y_n=x$, we obtain
\begin{align*}
h_{p+1,n}(i_1,\ldots,i_{p+1};x,\ldots,x)&=(n+1)h_{p,n+1}(i_1,\ldots,i_p;i_{p+1},x,\ldots,x)
\\&~~~~~~~~~~~~~\phantom{=}+h_{p,n-1}(i_1,\ldots,i_p;x,\ldots,x)(\fel{i_{p+1}}{}{x})\\&
=\frac{d}{dx_{p+1}}h_{p,n+1}(i_1,\ldots,i_p;x,\ldots,x)
\\&~~~~~~~~~~~~~\phantom{=}+h_{p,n-1}(i_1,\ldots,i_p;x,\ldots,x)(\fel{i_{p+1}}{}{x}),
\end{align*}
where $0<n<p$. The cases, when $n\in\{0,p,p+1\}$ can be considered separately, and have the following forms:
\begin{equation*}
h_{p+1,p+1}(i_1,\ldots,i_{p+1};x,\ldots,x)
=h_{p,p}(i_1,\ldots,i_p;x,\ldots,x)(\fel{i_{p+1}}{}{x}),
\end{equation*}
\begin{equation*}
h_{p+1,p}(i_1,\ldots,i_{p+1};x,\ldots,x)
=h_{p,p-1}(i_1,\ldots,i_p;x,\ldots,x)(\fel{i_{p+1}}{}{x}),
\end{equation*}
\begin{equation*}
h_{p+1,0}(i_1,\ldots,i_{p+1})
=\frac{d}{dx_{p+1}}h_{p,1}(i_1,\ldots,i_p;x).
\end{equation*}

Hence, summing separately the left and the right hand sides of the last four equalities, we make sure that the ansatz (\ref{form2}) is correct. The uniqueness of the solution is obvious. Thus the statement of the theorem is proved.
\end{proof}
\begin{lem}
Show that in all calculations of this section we can change $\fel{i}{}{j}$ by $\fel{i}{\beta}{j}$, where $\beta\in V$. Moreover, we can replace it with any symmetric matrix $g_{ij}$.
\end{lem}

\subsection{Physical point of view}
Here we assume $\beta=\mathbf{1}$.
Let us draw the attention on formula (\ref{phif}) and consider the following example 
\begin{equation}
\label{ve1}
F(x)=\sum_{k_1,\ldots,k_i\geqslant 1}C_{k_1\ldots k_i}x_{k_1}\cdot\ldots\cdot x_{k_i},\,\,\,\,\,\,x\in V,\,\,\,\,\,\,i\in\mathbb{N}.
\end{equation}

Hence, if we substitute the operator $\partial_{x}$ instead of $x$ into $F(x)$, and calculate the action (\ref{phif}) by applying Theorem \ref{th}, we obtain
\begin{equation}
\label{ve2}
\sum_{k_1,\ldots,k_i\geqslant 1}C_{k_1\ldots k_i}h_{i,0}(k_1,\ldots,k_i).
\end{equation}
\begin{definition}
	\label{def2}
The function $C_{k_1\ldots k_i}$ denotes $i$ lines, protruding from one point. 
The free ends of lines have indices $k_1,\ldots,k_i$.
Such element of diagram technique we call by vertex with $i$ lines.
\end{definition}

The last definition gives a possibility to visualize the transition from formula (\ref{ve1}) to (\ref{ve2}). It can be represented as follows: we connect the vertex indices with lines $\fel{k_n}{}{k_j}$ in all possible ways, and then sum over $\mathbb{N}^i$.
Moreover, if the kernel $C_{k_1\ldots k_i}$ can be factorized, then it is convenient to associate different factors with different vertices. So answer also contains vertices connected to each other.

Let us apply such construction to our main model. In the case we have $F_\alpha(x)=\exp\big\{i\sum_{k\geqslant1}x_k^3\alpha_k^{\phantom{3}}\big\}$, and, after the decomposition of the exponential into the series, all orders can be represented as products of $i\sum_{k\geqslant 1}x_k^3\alpha_k^{\phantom{3}}$. It means that the basic construct element has the form, depicted in Figure \ref{phi3}.

\begin{figure}[h]
	\center{\includegraphics{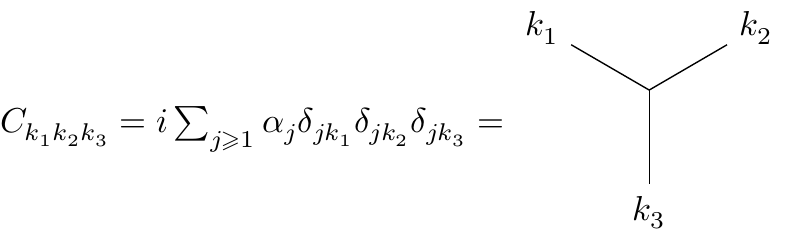}}
	\caption{Vertex for the cubic model.}
	\label{phi3}
\end{figure}

So it follows from the recipe that in order to get the answer for
\begin{equation}
\label{ex} \frac{1}{j!}\bigg(i\sum_{k\geqslant 1}\partial_{x_k}^3\alpha_k^{\phantom{3}}\bigg)^je^{(x,x)/2}\bigg|_{x=0},
\end{equation}
we need to draw the last diagram $j$ times, and then connect the external lines in all possible ways.

Now we should note some points. Firstly, in the process of drawing images we can omit indices. Secondly, if we have an odd number of external lines, then we have zero contribution. Hence, it is convenient to give definitions for an auxiliary operator.
\begin{definition}
\label{def3}
Let $\mathcal{D}$ be a diagram with $k$ external lines, then an operator $\mathbb{H}_{\mathrm{n}}$, where $0\leqslant n\leqslant k$, is defined by the following rules:
\begin{equation*}
\mathbb{H}_{\mathrm{n}}\mathcal{D}= 
\begin{cases}
\text{$k-n$ lines of $\mathcal{D}$ are connected in all possible ways,} &\text{if $k-n$ is even,}\\
\text{0,} &\text{if $k-n$ is odd.}
\end{cases}
\end{equation*}
In addition, if $n=0$ then the index can be omitted, and we have $\mathbb{H}_{\mathrm{0}}=\mathbb{H}$.
\end{definition}
The last operator is the linear one. Let us give some computations of formula (\ref{ex}). If the parameter $j$ is odd, then we have zero, because the diagram contains $3j$ external lines. If $j$ is even, then we have nonzero answer. For example, for $j=2$ we have

\begin{equation}\label{d1}
\frac{1}{2!}
\mathbb{H}
\Big(\includegraphics[scale=0.6]{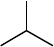}^2
\Big)=\frac{9}{2}~\och{-0.5}{0.6}{0} +3~ \lemm{-0.5}{0.8}{0}.
\end{equation}

\begin{definition}
\label{def4}
Diagram is called connected if any two vertices are connected by a path. Diagram is called strongly connected, or one-particle irreducible, if it remains connected after removing any one line, see \cite{christofides-1975}.
\end{definition}
\begin{definition}
\label{def5}
Let $n\in\mathbb{N}\cup\{0\}$.
An operator $\mathbb{H}^{\mathrm{c/sc}}_{\mathrm{n}}$ is defined in the same way as $\mathbb{H}_{\mathrm{n}}$, but besides keeps only connected/strongly connected diagrams, and removes other. Additionally, $\mathbb{H}^{\mathrm{c/sc}}1=0$.
\end{definition}

Let us give some examples:
\begin{equation}\label{d2}
\mathbb{H}
\Big(\includegraphics[scale=0.6]{ex}^j
\Big)=
\mathbb{H}^{\mathrm{c}}
\Big(\includegraphics[scale=0.6]{ex}^j
\Big)=
\mathbb{H}^{\mathrm{sc}}
\Big(\includegraphics[scale=0.6]{ex}^j
\Big)=0,\,\,\,\,\,\text{if $j$ is odd},
\end{equation}
\begin{equation}
\mathbb{H}^{\mathrm{c}}
\Big(\includegraphics[scale=0.6]{ex}^2
\Big)=
\mathbb{H}
\Big(\includegraphics[scale=0.6]{ex}^2
\Big),\,\,\,\,\,\,
\mathbb{H}^{\mathrm{sc}}
\Big(\includegraphics[scale=0.6]{ex}^2
\Big)=3!~ \lemm{-0.5}{0.8}{0},
\end{equation}
\begin{equation}
(\mathbb{H}-\mathbb{H}^{\mathrm{c}})
\Big(\includegraphics[scale=0.6]{ex}^4
\Big)=243~
\begin{smallmatrix}
\och{-0.5}{0.3}{0}\\
\och{-0.5}{0.3}{0}
\end{smallmatrix}
+324~
\begin{smallmatrix}
\lemm{-0.5}{0.4}{0}\\
\och{-0.5}{0.3}{0}
\end{smallmatrix}
+108~
\begin{smallmatrix}
\lemm{-0.5}{0.4}{0}\\
\lemm{-0.5}{0.4}{0}
\end{smallmatrix}.
\end{equation}

\begin{lemma}
\label{lem7}
Under the conditions described above we have the equality
\begin{equation}
\frac{1}{4!}
(\mathbb{H}-\mathbb{H}^{\mathrm{c}})
\Big(\includegraphics[scale=0.6]{ex}^4
\Big)=
\frac{1}{2}
\Bigg(
\frac{1}{2!}
\mathbb{H}
\Big(\includegraphics[scale=0.6]{ex}^2
\Big)\Bigg)^{2}.
\end{equation}
\end{lemma}
The last relation gives hypothesis for a formula that will be proved in the next section. Let us write it in first orders as the following chain of relations
\begin{align}
\nonumber
\ln\Big(\mathbb{H}\,
e^{\includegraphics[scale=0.6]{ex}}
\Big)&\approx\frac{1}{2!}\mathbb{H}
\Big(\includegraphics[scale=0.6]{ex}^2\Big)
+
\frac{1}{4!}\mathbb{H}\Big(\includegraphics[scale=0.6]{ex}^4\Big)
-\frac{1}{2}
\bigg(\frac{1}{2!}\mathbb{H}\Big(\includegraphics[scale=0.6]{ex}^2\Big)
+
\frac{1}{4!}\mathbb{H}\Big(\includegraphics[scale=0.6]{ex}^4\Big)\bigg)+\ldots \\\label{te}
&\approx\frac{1}{2!}\mathbb{H}^{\mathrm{c}}\Big(\includegraphics[scale=0.6]{ex}^2\Big)
+
\frac{1}{4!}\mathbb{H}^{\mathrm{c}}\Big(\includegraphics[scale=0.6]{ex}^4\Big)+\ldots \\
&\approx
\mathbb{H}^{\mathrm{c}}\,
e^{\includegraphics[scale=0.6]{ex}}.\nonumber
\end{align}
\begin{lem}
Check the diagram equalities mentioned above.
\end{lem}

\subsection{Calculation of $\ln\mathbb{H}\exp(\mbox{vertex})$}
\begin{theorem}
\label{th1}
Let $\zeta$ be a vertex with an arbitrary number of external lines. Then, using the physical kind of diagram technique, we have
\begin{equation}
\label{po}
\mathbb{H}\exp(\zeta)=\exp\big(\mathbb{H}^{\mathrm{c}}(\exp(\zeta))\big).
\end{equation}
\end{theorem}
\begin{proof}
Let $j$ be from $\mathbb{N}$. The first our step is to express $\mathbb{H}(\zeta^j)$ through a linear combination of connected diagram products of view $\mathbb{H}^{\mathrm{c}}(\zeta^{n_1})\cdot\ldots\cdot\mathbb{H}^{\mathrm{c}}(\zeta^{n_k})$, where $n_i\in\mathbb{N}$, $n_1+\ldots+n_k=j$, and $k\in\{1,\ldots,j\}$ is a number of connection components. 

It can be done in several stages.  Firstly, we should distribute $j$ vertices to $k$ groups. So we have multinomial coefficient $C^j_{n_1\ldots n_k}$. Then, we need to remove an internal symmetry, that follows from connection components with the same number of vertices ($n_i=n_j$). So we should divide the coefficient by a number $S_{id}(n_1,\ldots,n_k)$ of all identical permutations of the set $\{n_1,\ldots,n_k\}$. Finally we have the relation
\begin{equation}\label{po1}
\mathbb{H}(\zeta^j)=
\sum_{k=1}^{j}\,
\sum_{\substack{1\leqslant n_1\leqslant\ldots\leqslant n_k \\ n_1+\ldots+n_k=j}}
\frac{C^{j}_{n_1\ldots n_k}}{S_{id}(n_1,\ldots,n_k)}
\mathbb{H}^{\mathrm{c}}(\zeta^{n_1})\ldots\mathbb{H}^{\mathrm{c}}(\zeta^{n_k}).
\end{equation}

The second step is to decompose the right hand side of equality (\ref{po}). It also can be done in several stages. After applying the exponential definition we obtain
\begin{equation*}
\exp\big(\mathbb{H}^{\mathrm{c}}(\exp(\zeta))\big)=
1+\sum_{k\geqslant 1}\frac{1}{k!}\big(\mathbb{H}^{\mathrm{c}}(\exp(\zeta))\big)^k.
\end{equation*}

Then, we use the definition one more time
\begin{equation*}
\big(\mathbb{H}^{\mathrm{c}}(\exp(\zeta))\big)^k
=\sum_{n_1,\ldots,n_k\geqslant 1}
\frac{\mathbb{H}^{\mathrm{c}}(\zeta^{n_1})}{n_1!}\ldots\frac{\mathbb{H}^{\mathrm{c}}(\zeta^{n_k})}{n_k!},
\end{equation*}
and make the following resummations 
\begin{equation*}
\sum_{n_1,\ldots,n_k\geqslant 1}\to
\sum_{j\geqslant 0}\sum_{\substack{n_1,\ldots,n_k\geqslant1 \\ n_1+\ldots+n_k=j}}\to
\sum_{j\geqslant 0}\,\sum_{\substack{1\leqslant n_1\leqslant\ldots\leqslant n_k \\ n_1+\ldots+n_k=j}}S_{diff}(n_1,\ldots,n_k),
\end{equation*}
where $S_{diff}(n_1,\ldots,n_k)$ is a number of all different permutations of the set $\{n_1,\ldots,n_k\}$. Restoring the entire chain of calculations we get that the right hand side of equality (\ref{po}) has the form
\begin{equation*}
1+\sum_{j\geqslant 1}\frac{1}{j!}\Bigg(\sum_{k=1}^{j}\,\sum_{\substack{1\leqslant n_1\leqslant\ldots\leqslant n_k \\ n_1+...+n_k=j}} C^{j}_{n_1...n_k} \frac{S_{diff}(n_1,\ldots,n_k)}{k!}  
\mathbb{H}^{\mathrm{c}}(\zeta^{n_1})
\ldots\mathbb{H}^{\mathrm{c}}(\zeta^{n_k})\Bigg),
\end{equation*}

It is easy to see that the proof follows from the equality of (\ref{po1}) and the expression in brackets in the last formula. It means, that the proposition follows from the relation
\begin{equation}\label{q1}
k!=S_{id}(n_1,\ldots,n_k)S_{diff}(n_1,\ldots,n_k),
\end{equation}
which is obvious.
\end{proof}

\begin{lem}
Check that we can replace the vertex $\zeta$ in the theorem and in its proof with a finite sum of vertices.
\end{lem}

Let us consider one corollary. For that we introduce complete Bell polynomials, that can be defined by the following generating function
\begin{equation*}
1+\sum_{n\geqslant 1} A_{n}\left(x_{1}, x_{2}, \dots, x_n\right) \frac{t^{n}}{n !}=\exp \left[\sum_{k\geqslant 1} x_{k} \frac{t^{k}}{k!}\right],\,\,\,\mbox{and}\,\,\,A_{0}=1.
\end{equation*}

The polynomial $A_n$ has a combinatorial meaning, that shows how many possible partitions a set of $n$ elements has. For example, let us consider $A_{4}\left(x_{1}, x_{2}, x_{3}, x_{4}\right)=x_{1}^{4}+6 x_{1}^{2} x_{2}+4 x_{1} x_{3}+3 x_{2}^{2}+x_{4}$.
The coefficient 1 near the first term $x_1^4$ means, that there is only 1 way to split a set with four elements to four nonzero subsets. The second term $6x_1^2 x_2^{\phantom{2}}$ means, that there are 6 ways to split a set with four elements to three subsets of orders 1, 1, and 2. Other terms can be understood in the same way.

\begin{lemma} Let $\zeta$ be a vertex or a sum of few vertices with number of external lines, then for every $n \in \mathbb{N}\cup\{0\}$ we have
	\begin{equation*}
		\begin{split}
			\label{Hzeta}
			\mathbb{H} \left( \zeta^n \right)&=A_{n}\left(\mathbb{H}^{\mathrm{c}}\left(\zeta \right), \dots, \mathbb{H}^{\mathrm{c}}\left(\zeta^n \right) \right),\\
			\left(\mathbb{H}-\mathbb{H}^{\mathrm{c}} \right)\zeta^n&=A_n \left( \mathbb{H}^{\mathrm{c}} \left( \zeta \right), \dots, \mathbb{H}^{\mathrm{c}} \left( \zeta^{n-1} \right),0 \right).
		\end{split}
	\end{equation*}
	
\end{lemma}
\begin{proof} Let us introduce an auxiliary parameter $q$ into the relation from Theorem \ref{th1}.
Then we have the chain of equalities
\begin{align*}
\sum_{n\geqslant 0} \frac{\mathbb{H}(\zeta^n)}{n!} q^n&=
\mathbb{H}\exp(q\zeta)=\exp\big(\mathbb{H}^{\mathrm{c}}(\exp(q\zeta))\big)\\
&=\mathrm{exp} \left( \sum_{n\geqslant 0}\frac{\mathbb{H}^{\mathrm{c}} \left( \zeta^n \right) }{n!}q^n\right)=\sum_{n\geqslant 0}\frac{A_{n}\left(\mathbb{H}^{\mathrm{c}} \left(\zeta \right), \dots, \mathbb{H}^{\mathrm{c}} \left(\zeta^n \right) \right)}{n!}q^n,
\end{align*}
from which the first statement follows.

The second relation follows from the first one and from the polynomial property: for $n \geqslant 1$ we have $A_n\left(x_1, \dots, x_n \right)=A_n\left( x_1, \dots, x_{n-1},0 \right)+x_n$.
\end{proof}
Considering $A_4$, it is easy to check the relation from Lemma \ref{lem7}.

\section{Generating functions}\label{gen}
\subsection{Definitions}
In the last sections we discovered the functional with cubic potential, which is actually an infinite product of integrals $I(\alpha)$, see (\ref{intI}). Now we are going to expand our definition in the following way
\begin{equation}
\label{prf1}
Z(\beta,\alpha)=\frac{1}{\sqrt{2\pi}}\int_{\mathbb{R}}e^{-\frac{x^2}{2}}e^{i\alpha x^3}e^{\beta x}\,dx=\sum_{k\geqslant 0}\frac{\beta^k}{k!}Z_k(\alpha),
\end{equation}
where $\alpha\in\mathbb{R}$, $\beta\in\mathbb{C}$, and we have used the Taylor expansion of the exponential.
\begin{lem}
Write an explicit formula for $Z_k(\alpha)$.
\end{lem}
The next function can be obtained from the previous one by using the logarithm operation
\begin{equation}
\label{prf2}
W(\beta,\alpha)=\mathrm{ln}\big(Z(\beta,\alpha)\big)=\sum_{k\geqslant 0}
\frac{\beta^k}{k!}W_k(\alpha).
\end{equation}
\begin{lem}\label{d3}
Check the validity of the following decompositions:
\begin{align*}
Z(0,\alpha)&=1-\frac{15}{2}\alpha^2+\frac{3465}{8}\alpha^4+O(\alpha^6), \\
W(0,\alpha)&=-\frac{15}{2}\alpha^2+405\alpha^4+O(\alpha^6).
\end{align*}
\end{lem}

A third functional has a little bit complicated definition, because we should do Legendre transformation. So firstly we consider a differential equation of the form
\begin{equation}
\label{eq}
\theta=\frac{dW(\alpha,\beta)}{d\beta},
\end{equation}
where $\theta$ is an auxiliary parameter. Of course, a solution exists not for all values of the parameter $\theta$, so we assume that it has a suitable good value. Then we suppose, that we have solved the equation with respect to $\beta=\beta(\theta,\alpha)$, which is the function of the parameters $\theta$ and $\alpha$.

Now we are ready to define the last function by the formula
\begin{equation}
\label{gadef}
\Gamma(\theta,\alpha)=W(\beta(\theta,\alpha),\alpha)-\theta\beta(\theta,\alpha),
\end{equation}
that has the following series decomposition in powers of the parameter $\theta$
\begin{equation}
\label{prf3}
\Gamma(\theta,\alpha)=\sum_{k\geqslant 0}\frac{\theta^k}{k!}\Gamma_k(\alpha).
\end{equation}

We see that Legendre transformation includes the transition from the variables $(\beta,\alpha)$ to $(\theta,\alpha)$. Also we note some particular values:
\begin{equation}
\label{partic}
Z(0,0)=1,\,\,\,\,W(0,0)=0,\,\,\,\,\Gamma(0,\alpha)=W(\beta(0,\alpha),\alpha).
\end{equation}

\subsection{Equations for Z}
\begin{lemma}
\label{lem4}
Let $\alpha\in\mathbb{R}$ and $\beta\in\mathbb{C}$. Then the functional $Z(\beta,\alpha)$ satisfies the equations
\begin{align}
\label{Zab}
\Big[(1-3i\alpha\beta)\partial_\beta-i9\alpha^2\partial_\alpha-
(3i\alpha+\beta)\Big]Z(\beta,\alpha)&=0,
\\
\label{dshZ}
\Big[3i\alpha\partial^2_\beta-\partial_{\beta}^{\phantom{2}}+\beta\Big]Z(\beta,\alpha)&=0.
\end{align}
\end{lemma}
\begin{proof} Let us introduce a set of auxiliary functions as follows
\begin{equation*}
Z_n=\frac{1}{\sqrt{2\pi}}\int_{\mathbb{R}}x^n e^{-\frac{x^2}{2}}e^{i\alpha x^3}e^{\beta x}\,dx,
\end{equation*}
where $n\in\mathbb{N}\cup\{0\}$ and $Z_0=Z(\beta,\alpha)$. Using the integration by parts we get that the functions satisfy the relations
\begin{equation*}
\label{Z_n}
Z_{n}=(n-1)Z_{n-2}+3i\alpha Z_{n+1}+\beta Z_{n-1},\,\,\,\,\,\,n\geqslant 1,
\end{equation*}
where $Z_{-1}=0$.
So we can write the following chain of equalities
\begin{equation*}
Z_1=3i\alpha Z_2+\beta Z_0=(\beta+3i\alpha)Z_0+3i\alpha\beta Z_1-9\alpha^2 Z_3,
\end{equation*}
from which, after using 
$\partial_\alpha Z_0=i Z_3$, and $\partial_\beta Z_0=Z_1$, the first equality follows. The second equality can be obtain by integration by parts in the same manner.
\end{proof}
\begin{definition}
\label{def6}
Let us introduce an action in the form $S[x]=-x^2/2+i\al x^3$. Then the second equation from Lemma \ref{lem4} can be rewritten in the form
\begin{equation}\label{act}
\Big[S'[\partial_\beta]+\beta\Big]Z(\beta,\alpha)=0,
\end{equation}
and is called Dyson--Schwinger equation, see \cite{vasiliev-1998}.
\end{definition}

\begin{figure}[h]
	\center{\includegraphics[width=0.3\linewidth]{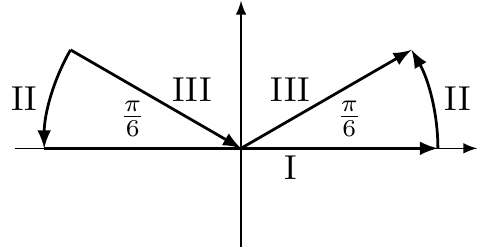}}
	\caption{Contours for Cauchy integral theorem.}
	\label{pic2}
\end{figure}

Let us express the function $Z(\beta,\alpha)$ in terms of known special functions. For this purposes in integral (\ref{prf1}) we do change of variable in the form $\alpha=z(3\alpha)^{-1/3}+(6i\alpha)^{-1}$. Then we have
\begin{equation}
Z(\beta,\alpha)=\frac{1}{\sqrt{2\pi}}
\frac{e^{\frac{1}{108 \alpha^2}-i\frac{\beta}{6\alpha}}}{(3\alpha)^{1/3}}
\int_{\mathbb{R}}e^{i(\frac{z^3}{3}+pz)}\,dz,
\end{equation}
where $p=(3\alpha)^{-1/3}\big(\frac{1}{12 \alpha}-i\beta\big)$. After that, using the Cauchy integral theorem (see \cite{walker-1975}) for contours in Figure \ref{pic2}
and the definition for Airy functions 
\begin{equation}
\int_{\mathrm{I}}e^{i(\frac{z^3}{3}+pz)}dz=
\int_{\mathrm{III}}e^{i(\frac{z^3}{3}+pz)}dz=2\pi\mathrm{Ai}(p),
\end{equation}
we obtain the final formula
\begin{equation}
Z(\beta,\alpha)=\sqrt{2\pi}
\frac{e^{\frac{1}{108 \alpha^2}-i\frac{\beta}{6\alpha}}}{(3\alpha)^{1/3}}
\mathrm{Ai}(p).
\end{equation}
\begin{lem}
Show that it is possible to apply the Cauchy integral theorem.
\end{lem}

\subsection{Equations for W}
\begin{lemma}
\label{lem5}
Let $\alpha\in\mathbb{R}$ and $\beta\in\mathbb{C}$. Then the functional $W(\beta,\alpha)$ satisfies the equations
\begin{equation}
\label{baW}
\Big[(1-3i\alpha\beta)\partial_\beta-i9\alpha^2\partial_\alpha\Big]W(\beta,\alpha)-(3i\alpha+\beta)=0,
\end{equation}
\begin{equation}
\label{bW}
3i\alpha\Big[(\partial^{\phantom{2}}_\beta W(\beta,\alpha))^2+\partial^{2}_\beta W(\beta,\alpha)\Big]-\partial^{\phantom{2}}_\beta W(\beta,\alpha)+\beta=0.
\end{equation}
\end{lemma}
\begin{proof}
Both equations follows from the equalities of Lemma \ref{lem4} by substitution $Z=\exp(W)$ and using the relation
$Z^{-1}\partial^{2}_\beta Z=(\partial^{\phantom{2}}_\beta W)^2+\partial^{2}_\beta W$.
\end{proof}

\subsection{Functions $\Gamma_0$ and $\Gamma_1$}
\begin{lemma}
\label{gammas}
Let $\Gamma(\theta,\alpha)$, $\Gamma_0(\alpha)$, and $\Gamma_1(\alpha)$ be from decomposition (\ref{prf3}). Let also
\begin{equation}
\label{bet}
\tilde{\beta}(\alpha)=\beta(0,\alpha),\,\,\,\,\,\,
w(\alpha)=W(\tilde{\beta}(\alpha),\alpha).
\end{equation}
Then we have 
\begin{equation}
\label{gam}
\Gamma_0(\alpha)=w(\alpha),\,\,\,\,\,\,
\Gamma_1(\alpha)=-\tilde{\beta}(\alpha).
\end{equation}
\end{lemma}
\begin{proof}
We are going to use the Legendre transformation in the form (\ref{gadef}). The first equality follows from it after substitution $\theta=0$ and using the relation $\Gamma_0(\alpha)=\Gamma(0,\alpha)$ from (\ref{prf3}). To obtain the second equality we should consider the differentials
\begin{align*}
d\Gamma(\theta,\alpha)&=\pp{W(\beta,\alpha)}{\beta}d\beta(\theta,\alpha)+\pp{W(\beta,\alpha)}{\alpha}d\alpha-\be(\theta,\alpha)d\theta-\theta d\be(\theta,\alpha),\\
\beta(\theta,\alpha)&=\pp{\beta(\theta,\alpha)}{\theta}d\theta+\pp{\beta(\theta,\alpha)}{\alpha}d\alpha,
\end{align*}
from which, by using equation (\ref{eq}), new formula follows
\begin{equation}
d\Gamma(\theta,\alpha)=-\beta(\theta,\alpha)d\theta+\pp{W(\beta,\alpha)}{\alpha}d\alpha\Rightarrow\frac{\partial\Gamma(\theta,\alpha)}{\partial\theta}=-\beta(\theta,\alpha).
\end{equation}

To see lemma statement we need to note that $\Gamma_1(\alpha)=\frac{\partial\Gamma(\theta,\alpha)}{\partial\theta}\big|_{\theta=0}
=-\tilde{\beta}(\alpha)$.
\end{proof}
\begin{lemma}
Under the conditions described above we have
\begin{equation}
\label{watib}
9i\alpha^2\pp{w(\alpha)}{\alpha}+3i\alpha+\tib(\alpha)=0,
\end{equation}
\begin{equation}
\label{abw}
\bigg(1-3i\alpha\tilde{\beta}+9i\alpha^2\pp{\tib}{\alpha}\bigg)\ppp{W(\tilde{\beta},\alpha)}{\tilde{\beta}}-1=0.
\end{equation}
\end{lemma}
\begin{proof}
Let us decompose the function $W(\beta,\alpha)$ in a Taylor series in powers of $(\beta-\tilde{\beta})$ near the point $\tilde{\beta}$. So we have
\begin{equation*}
W(\beta,\alpha)=\sum_{k\geqslant 0}\frac{(\beta-\tilde{\beta})^k}{k!}\frac{\partial^k W(\beta,\alpha)}{\partial \beta^k}\bigg|_{\beta=\tilde{\beta}}= \sum_{k\geqslant 0} \frac{(\beta-\tilde{\beta})^k}{k!}\frac{\partial^k W(\tilde{\beta},\alpha)}{\partial \tilde{\beta}^k}.
\end{equation*}

If we substitute the last equality into equation (\ref{baW}) 
\begin{equation*}
\sum_{k\geqslant 0}\frac{(\beta-\tilde{\beta})^{k}}{k!}
\bigg[\bigg(1-3i\alpha\tilde{\beta}+9i\alpha^2\pp{\tilde{\beta}}{\alpha}\bigg)
\frac{\partial}{\partial\tilde{\beta}}-9i\alpha^2\frac{\partial}{\partial\alpha}-3ik\alpha
\bigg]\frac{\partial^kW(\tilde{\beta},\alpha)}{\partial\tilde{\beta}^k}
-3i\alpha-\tilde{\beta}-(\beta-\tilde{\beta})=0,
\end{equation*}
we obtain a number of recurrent relations for the series coefficients. So the lemma statement follows from the zero and first orders by using the equality $\partial W(\tilde{\beta},\alpha)/\partial\tilde{\beta}=0$, see formula (\ref{eq}), when $\theta=0$.
\end{proof}

\subsection{Solutions for $\tilde{\beta}$ and $w(\alpha)$}
\begin{lemma}
\label{lem6}
The function $\tilde{\beta}$ from equality (\ref{bet}) satisfies the first order equation 
\begin{equation}
\label{eq1}
9i\alpha^2\tib(\alpha)\partial_\alpha\tib(\alpha)-3i\alpha\tib^2(\alpha)+3i\alpha+\tib(\alpha)=0
\end{equation}
with the initial condition $\tilde{\beta}(0)=0$. Moreover, near the point $\alpha=0$ it can be represented by the series 
\begin{equation}
\label{eq2}
\tib(\alpha)=\sum_{k\geqslant 0}b_{2k+1}\alpha^{2k+1},
\end{equation}
the coefficients of which satisfy the relations
\begin{equation}
\label{bcoef}
b_1=-3i,\,\,\,\,\,\,
b_{2k+3}=-6i\sum_{n=0}^{k}\lr{3\lr{k-n}+1}b_{2n+1}b_{2k-2n+1},\,\,\,\,\,\,
k \geqslant 0.
\end{equation}
\end{lemma}
\begin{proof}
Equation (\ref{eq1}) can be obtained by substitution of relation (\ref{bW}) into equality (\ref{abw}), where it was used that at point $\theta=0$ we have 
\begin{equation*}
\frac{\partial W(\tilde{\beta},\alpha)}{\partial\tilde{\beta}}=0,\,\,\,\,\,\,
3i\alpha\frac{\partial^2 W(\tilde{\beta},\alpha)}{\partial\tilde{\beta}^2}+\tib(\alpha)=0.
\end{equation*}

Then if we substitute series (\ref{eq2}) into equation (\ref{eq1}) we obtain
\begin{equation*}
\sum_{k\geqslant 0}\bigg(\sum_{n=0}^k6i (3(k-n)+1)b_{2n+1}b_{2k-2n+1}+b_{2k+3}\bigg)\alpha^{2k+3}+(b_1+3i)\alpha=0,
\end{equation*}
from which the last statement of lemma follows.
\end{proof}
\begin{lemma}
\label{lem8}
The function $w(\alpha)$ from (\ref{bet}) near the point $\alpha=0$ can be represented by the asymptotic series $w(\alpha)=\sum_{k\geqslant 0}c_{2k}\alpha^{2k}$, the coefficients of which satisfy the equalities
\begin{equation}
c_0=0,\,\,\,\,\,\,c_{2k}=-\frac{b_{2k+1}}{18 i k},\,\,\,\,\,\,k \geqslant 1\,
\end{equation}
where $b_{2k+1}$ is from (\ref{bcoef}).
\end{lemma}
\begin{proof}
The statement follows after substitution of the series for $w(\alpha)$ into equation (\ref{watib}) and using representation (\ref{eq2}) for function $\tib(\alpha)$, where the initial condition $w(0)=W(\tib(0),0)=0$ was also used.
\end{proof}
\begin{lem}
	Check the validity of the following decomposition:
	\begin{align*}
	\Gamma(0,\alpha)=-3\alpha^2+135\alpha^4+O(\alpha^6).
	\end{align*}
\end{lem}

\section{Result interpretation}
\label{res}

\subsection{Functions $Z_0$ and $W_0$}
In the previous section, the functions $Z$, $W$, and $\Gamma$ were studied by using mathematical language. However, a graphical interpretation is also possible. Consider the function $Z_0=Z(0,\alpha)$, or rather its series
\begin{equation}\label{z}
Z_0\left(\alpha \right)=\sum_{p\geqslant 0}\frac{(-1)^p}{8^p} 
\frac{(6p)!}{(2p)!(3p)!} \alpha^{2p} =  1-\frac{15}{2} \alpha^2+O(\alpha^4).
\end{equation}

Let us note that the diagram \includegraphics[scale=0.6]{ex} is proportional to $i\alpha$. Then, using the above calculations (\ref{d1}) and (\ref{d2}), we can write the approximation
\begin{align*}
\mathbb{H}\left(1+ \includegraphics[scale=0.6]{ex} +\frac{1}{2!}\includegraphics[scale=0.6]{ex}^2 +\dots \right)
&=1+ \frac{9}{2}~ \och{-0.65}{0.7}{0}+3 ~\lemm{-0.65}{0.7}{0}+ \dots=\\
&=1+\frac{9}{2}\left(i \alpha \right)^2+3\left(i \alpha \right)^2+\dots=\\
&=1-\frac{15}{2} \alpha^2+ \dots
\end{align*}

We see that the row coincides with one from (\ref{z}). It turns out that to obtain the expansion coefficient of $Z$ in the Taylor series, you need to draw all possible diagrams in the appropriate order, then match $1$ to each line, match $i \alpha$ to each vertex, and divide by the factorial. This can be written as follows.

\begin{theorem} 
Under the conditions described above, we have the equality
\begin{equation}
Z_0 \left( \alpha \right)=\mathbb{H}\left( e^{\includegraphics[scale=0.6]{ex}}   \right) .
\end{equation}
\end{theorem}
\begin{proof}
Let us consider the order of $2p$, where $p\in\mathbb{N}$, then we can write
\begin{align*}
\frac{1}{(2p)!} \mathbb{H}\left( \includegraphics[scale=0.6]{ex}^{2p}\right)&=\frac{1}{(2p)!} \left({\substack{\text{All possible vacuum}\\ \text{diagrams with $2p$ vertices}}} \right)=\\
&=\frac{1}{(2p)!} \left({\substack{\text{Number of all possible vacuum}\\ \text{diagrams with $2p$ vertices} }} \right)\left(i \alpha \right)^{2p}.
\end{align*}

Let us count the number of diagrams. If there are $2p$ vertices, then there are $6p$ lines. 
It means that we should calculate a number of ways to split $6p$ free lines into pairs. It equals to $(6p-1)!!$. Hence, we have
 \begin{equation*}
\frac{1}{(2p)!} \mathbb{H}\left( \includegraphics[scale=0.6]{ex}^{2p} \right)=
(-1)^p \frac{(6p-1)!!}{(2p)!} \alpha^{2p}=
\frac{(-1)^p}{8^p} \frac{(6p)!}{(2p)!(3p)!} \alpha^{2p},
\end{equation*}
from which the statement follows.
\end{proof}

\begin{lem}
Verify that the equality $Z_n \left( \alpha \right)=\mathbb{H}_{\mathrm{n}}\left( e^{\includegraphics[scale=0.6]{ex}} \right)$ holds for all $n\in\mathbb{N}\cup\{0\}$.
\end{lem}

Similar properties are observed in the function $W$. You can make sure that in the first orders, $W_0$ is the sum of only connected diagrams.

\begin{theorem} Under the conditions described above, we have the equality
\begin{equation*}
W_0\left( \alpha \right)=\mathbb{H}^{\mathrm{c}}\left( e^{\includegraphics[scale=0.6]{ex}}   \right).
\end{equation*}
\end{theorem}
\begin{proof} Let us use Theorem \ref{th1}. Hence, we obtain the chain of equalities
\begin{equation*}
W_0=\mathrm{ln} \left( Z_0 \right)= \mathrm{ln} \left(\mathbb{H} \left(e^{\includegraphics[scale=0.6]{ex}}\right) \right)=\mathbb{H}^{\mathrm{c}}\left(e^{\includegraphics[scale=0.6]{ex}}  \right).
	\end{equation*}
\end{proof}

Now we give the standard notations for connected diagrams 
\begin{equation}
W_0= \mg{-0.6}{1}{0.7}\,,\,\,\,
W_1= \gl{-0.6}{1}{0.7}{0}\,,\,\ldots
\end{equation}
where $W_k$ has $k$ external lines.

\subsection{Functions $\Gamma_0$ and $\Gamma_1$}
\paragraph{Equation (\ref{watib}).}
Let us take a closer look at equation (\ref{watib}) and try to interpret it by using language of diagrams. To do this, we rewrite the equation in a more convenient form. Using the equalities $w\left(\alpha\right)=\Gamma_0$ and $\tilde{\beta}\left(\alpha\right) =-\Gamma_1$ from Lemma \ref{gammas}, we have
 \begin{equation}
\label{gh1}
6 \cdot i \alpha \cdot \frac{3}{2}i \alpha \frac{\partial}{\partial i \alpha} \Gamma_0+3 \cdot i\alpha=\Gamma_1.
\end{equation}

We are going to show that $\Gamma_0$ is all vacuum strongly connected diagrams, and $\Gamma_1$ is all strongly connected diagrams with one external line.  For convenience they are denoted as follows 

\begin{equation}\label{d}
\mathbb{H}^{\mathrm{sc}}\left(e^{\includegraphics[scale=0.6]{ex}}  \right)=
\mlg{-0.6}{1}{0.7}\,,~~~~~~~~~
\mathbb{H}^{\mathrm{sc}}_1\left(e^{\includegraphics[scale=0.6]{ex}}  \right)=
\gol{-0.6}{1}{0.7}{0}\,.
\end{equation}

To begin with, we note that the operator $x \partial_x $ acting on $x^n$ gives $n x^n$. Let us find a solution for $\Gamma_0$ in the form of combination of some vacuum diagrams. Using the fact that in our case all vacuum diagrams are polynomials of the parameter $i\alpha$, since the vertex is proportional to $i\alpha$, we obtain for an arbitrary term $\gamma_0$ from $\Gamma_0$
\begin{equation}
i \alpha \frac{\partial }{i \alpha} \gamma_0=\left(\text{number of vertices in $\gamma_0$} \right) \cdot \gamma_0.
\end{equation}

If the number of vertices in the diagram $\gamma_0$ is multiplied by 3/2, then we get number of internal lines, because each vertex contains three lines, and each line has two vertices at the ends. That is
\begin{align}
\frac{3}{2}  i \alpha \frac{\partial }{i \alpha} \gamma_0=\left(\text{number of internal lines in $\gamma_0$} \right) \cdot \gamma_0.
\end{align}

Further, the factor $6\cdot i\alpha$ symbolizes the number of ways to connect the triple vertex
\includegraphics[scale=0.6]{ex} to two lines. In other words, we take $\gamma_0$, cut one inner line in an arbitrary way, and connect a triple vertex to the cut. And then we summarize all such diagrams for all internal lines.
\begin{definition}
\label{def7}
Let $D$ be a diagram and an operator $\mathbb{F}(D)$ have one argument. Then, the operator acts on $D$ as follows: it cuts an arbitrary internal line in $D$, attaches the triple vertex \includegraphics[scale=0.6]{ex} in all possible ways, and then sums over all the internal lines of $D$.
\end{definition}

Described above conditions mean, that we need to find such diagram sums $\Gamma_0$ and $\Gamma_1$ to equality (\ref{gh1}) holds in the form
\begin{equation}
\label{ex8}
\mathbb{F}(\Gamma_0)+3~\chups{-0.61}{0.53}{0.75}{0}
=\Gamma_1\,.
\end{equation}

It is easy to verify that the functions 
from definition (\ref{d}) satisfy the last equality. But they are not unique. So to fix the answer we need to analyze equation (\ref{eq1}).
\begin{lem}
Verify that the following equality holds
$\mathbb{F}\Big(\mlg{-0.61}{1}{0.7}\Big)+3~\chups{-0.61}{0.53}{0.75}{0}
=\gol{-0.61}{1}{0.7}{0}\,$.
\end{lem}

\paragraph{Equation (\ref{eq1}).}
Now we move on to decode equation (\ref{eq1}). Let us use  formula (\ref{gam}) and rewrite (\ref{eq1}) in a more convenient form
\begin{align}\label{dd}
3i\alpha=\Gamma_1\left(1+3i\alpha\Gamma_1-9i\alpha^2\frac{\partial}{\partial\alpha} 
\Gamma_1\right).
\end{align}

In the previous section we have chosen the ansatz for $\Gamma_0$. According to the ansatz, the function $\Gamma_1$ is a set of diagrams with one external line.
Let us consider the last term of (\ref{dd}) for an arbitrary term $\gamma_1$ from $\Gamma_1$ in more detail. Using the arguments from the previous section, we can write out the chain of equalities
\begin{align}
9  i \alpha\bigg( i \alpha \frac{\partial}{\partial i \alpha}\bigg) \gamma_1&\nonumber
=9 i \alpha \big[\text{number of vertices in $\gamma_1$}\big]\cdot\gamma_1\\\label{dd1}
&=6 i \alpha \big[ \text{number of internal lines in $\gamma_1$}\big]\cdot\gamma_1 + 3 i \alpha \gamma_1\\
&=\mathbb{F}(\gamma_1) + 3 i \alpha \gamma_1.\nonumber
\end{align}

Further, after summing by all $\gamma_1$ we substitute the last relation into equality (\ref{dd}). Hence, we have
\begin{align}\label{dd2}
3i\alpha=\Gamma_1\cdot\big[1-\mathbb{F}(\Gamma_1)\big]
\Longrightarrow
\Gamma_1=3i\alpha\sum_{k\geqslant 0}\Big(\mathbb{F}(\Gamma_1)\Big)^k.
\end{align}

Let us prove one auxiliary relation and apply it to our situation.
\begin{lemma}\label{dd6}
Let us consider the vertex with three external lines \includegraphics[scale=0.6]{ex} and define the following object
\begin{equation}\label{dd3}
\pov{-0.57}{1}{0.7}=\mathbb{H}^{\mathrm{sc}}_2\left(e^{\includegraphics[scale=0.6]{ex}}  \right).
\end{equation}

Then, under the conditions described above, we have the equality, depicted in Figure \ref{pic}.
\begin{figure}[h]
\center{\includegraphics{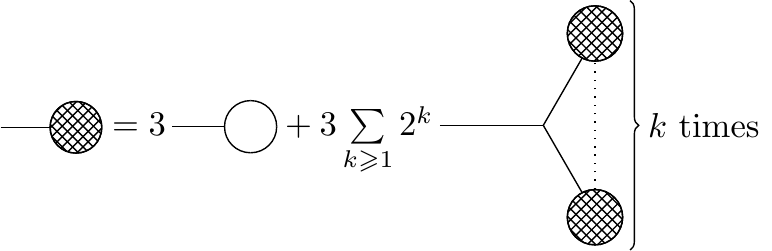}}
\caption{Decomposition of strongly connected diagrams.}
\label{pic}
\end{figure}
\end{lemma}
\begin{proof} Let us start from the left hand side of the equality, drawn in Figure \ref{pic}, and use definition (\ref{d}) for the sum of strongly connected diagrams with one external line.
After that we can choose one vertex, one end of which will remain external. Of course, we should take into account the corresponding coefficient, that follows from selecting one element among $n+1$ identical ones. So we get the following chain of equalities
\begin{equation}\label{dd4}
\gol{-0.6}{1}{0.7}{0}=\mathbb{H}^{\mathrm{sc}}_1\left(
e^{\includegraphics[scale=0.6]{ex}}  \right)
=\sum_{n\geqslant 0}\frac{1}{(n+1)!}\mathbb{H}^{\mathrm{sc}}_1\left(
\includegraphics[scale=0.6]{ex}^{n+1}  \right)=
3\,\chups{-0.6}{0.53}{0.7}{0}+
\sum_{n\geqslant 1}\frac{1}{n!}\mathbb{H}^{\mathrm{sc}}_1\Big(
\includegraphics[scale=0.6]{ex}\,\,\mathbb{H}^{\mathrm{c}}_2\left(
\includegraphics[scale=0.6]{ex}^n \right) \Big).
\end{equation}

Then, we are going to use the fact that any connected diagram $\zeta$ from $\mathbb{H}^{\mathrm{c}}_2\left(
\includegraphics[scale=0.6]{ex}^n \right)$ can be represented as a tree after all strongly connected subdiagrams are replaced by dots. Moreover, we are interested only in such $\zeta$ that $\mathbb{H}(\zeta)$ is a strongly connected diagram. It means that we should study only chains. Hence, taking into account necessary factors, the last part of (\ref{dd4}) can be rewritten in the form
\begin{equation}
3\,\chups{-0.6}{0.53}{0.7}{0}+
\sum_{n\geqslant 1}
\sum_{k\geqslant 1}
\sum_{\substack{n_k\geqslant\ldots\geqslant n_1\geqslant 1 \\ n_k+\ldots+n_1=n}}
\frac{C^n_{n_k\ldots n_1}}{n!S_{id}(n_k,\ldots,n_1)}\mathbb{H}^{\mathrm{sc}}_1\Big(
\includegraphics[scale=0.6]{ex}\,\, 
\mathbb{H}^{\mathrm{sc}}_2\left(\includegraphics[scale=0.6]{ex}^{n_k}\right)\cdot\ldots\cdot 
\mathbb{H}^{\mathrm{sc}}_2\left(\includegraphics[scale=0.6]{ex}^{n_1}\right) \Big).
\end{equation}

Then, applying the following resummation
\begin{equation}
\sum_{n\geqslant 1}
\sum_{\substack{n_k\geqslant\ldots\geqslant n_1\geqslant 1 \\ n_k+\ldots+n_1=n}}
\frac{C^n_{n_k\ldots n_1}}{n!S_{id}(n_k,\ldots,n_1)}
\to
\sum_{n_k,\ldots,n_1\geqslant 1}\frac{1}{k!n_k!\cdot\ldots\cdot n_1!},
\end{equation}
relation (\ref{q1}), and definition (\ref{dd3}), we obtain the formula
\begin{equation}
3\,\chups{-0.6}{0.53}{0.7}{0}+
\sum_{k\geqslant 1}
\frac{1}{k!}\mathbb{H}^{\mathrm{sc}}_1\left(
\includegraphics[scale=0.8]{ex}\,\,\left(\pov{-0.57}{1}{0.7}\right)^k \right),
\end{equation}
from which the statement of the lemma follows.
\end{proof}
\begin{lem}\label{dd5}
Verify that the following equality holds
$\pov{-0.61}{1}{0.7}=\frac{1}{2}
\mathbb{F}\Big(\gol{-0.61}{1}{0.7}{0}\Big)$.
\end{lem}

Now we are ready to apply all relations described above to our problem. Indeed, by using the definition from Exercise \ref{dd5}, we can rewrite the right hand side of the equality depicted 
in Figure \ref{pic} in the form drawn in Figure \ref{picpic}.
\begin{figure}[h]
	\center{\includegraphics{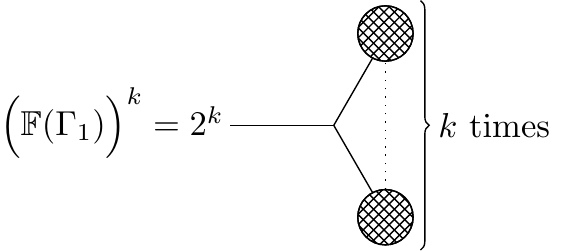}}
	\caption{Relation for the simplest cubic model.}
	\label{picpic}
\end{figure}

The last relation leads to the connection between the equation from Lemma \ref{dd6} and the equation from (\ref{dd2}). Hence, we obtain the following solutions
\begin{equation}
\Gamma_0=
\mlg{-0.6}{1}{0.7}\,,~~~~~~~~~
\Gamma_1=
\gol{-0.6}{1}{0.7}{0}\,.
\end{equation}

\section{Background field method}\label{back}
\subsection{The main idea}
Let us note, that due to the factorization (\ref{intI}) of the cubic model we can study only one-dimensional integral. The main object of this section is $I(\alpha)$ from the second formula of (\ref{intI}). The logarithm of $I(\alpha)$ is called an effective action 
\begin{equation}
J\left(\alpha \right)= \mathrm{ln}\left(\frac{1}{\sqrt{2 \pi}} \int e^{S[x]} dx\right),
\end{equation}
and plays an important role in physics, because it equals to sum of the classical action $S[y]$ from Definition \ref{def6} at some point $y$ and quantum corrections. Let us derive this representation.

First of all we shift the integration variable $x \rightarrow x+y$ and then scale it as follows $x\to x/\sqrt{1-6 i \alpha y}$. After that we obtain the decomposition
\begin{align}
	\label{fr}
	J\left(\alpha \right)= S[y]-\frac{1}{2}\mathrm{ln}(1-6i \alpha y )+\mathrm{ln}\left(\frac{1}{\sqrt{2 \pi}} \int e^{-\frac{x^2}{2}+is(\alpha,y)x^3 +t(\alpha,y)x} dx\right),
\end{align}
where
\begin{equation*}
t=t(\alpha,y)=\frac{-y+3i\alpha y^2}{(1-6 i \alpha y)^{1/2}},\,\,\,\,\,\,
s=s(\alpha,y)=\frac{\alpha}{(1-6 i \alpha y)^{3/2}}.
\end{equation*}

From formula (\ref{fr}) we see, that the quantum action $J\left(\alpha \right)$ equals to sum of the classical action $S[y]$, the first correction $-\frac{1}{2}\mathrm{ln}(1-6i \alpha y)$, and the high corrections.

We note, that  a free parameter $y$ appeared in the integral. It can be selected in a convenient way. Let us pay attention to the third term (with a logarithm)
\begin{equation*}
W(t,s)=W\bigg(\frac{-y+3i\alpha y^2}{(1-6 i \alpha y)^{1/2}},\frac{\alpha}{(1-6 i \alpha y)^{3/2}}\bigg).
\end{equation*}

Then, we are going to use the Lemma \ref{gammas}, that allows us to interpret $W(t,s)$ as the sum of strongly connected diagrams.  To do this, let us select the first argument of the function $W$ equal to the value of the function $\tilde{\beta}$ at the point $s$. This transition allows you to reduce the number of calculations in the general case.

Thus, we need to solve the nonlinear equation $t(\alpha,y)=\tilde{\beta}(s(\alpha,y))$. Let us use representation (\ref{eq2}) for the function $\tilde{\beta}$  and write the relation explicitly
\begin{equation} \label{eqfory}
	\frac{-y+3i\alpha y^2}{(1-6 i \alpha y)^{1/2}}=\sum_{k\geqslant 0} b_{2k+1} \left( \frac{\alpha}{(1-6 i \alpha y)^{3/2}}\right)^{2k+1} .
\end{equation}

Actually we need to find the solution $y$ of this equation as a function of the parameter $\alpha$. Then, substituting it into the formula (\ref{fr}), we get the value of the integral $J(\alpha)$. At the same time, if we consider only the first few strongly connected contributions to $W$, we get an approximate value of the integral. Note that for the calculation of strongly connected diagrams we must take into account, that each line corresponds to $1$, and each vertex
\includegraphics[scale=0.6]{ex}
corresponds to $i\alpha/(1-6i\alpha y_0)^{3/2}$.

In other words, using Lemmas \ref{lem6} and \ref{lem8}, we can write out the third term of formula (\ref{fr}) in the form
\begin{equation*}
W(\tilde{\beta}(s),s)=\Gamma_0(s)=\sum_{k\geqslant 1} c_{2k}s^{2k}=
\sum_{k\geqslant 1} c_{2k} \left( \frac{\alpha}{(1-6 i \alpha y)^{3/2}}\right)^{2k}.
\end{equation*}

\subsection{Example}
In this section we are going to find an approximation for the solutions of equation (\ref{eqfory}). We search for the solution in the form of a Taylor series $y(\alpha)=\sum_{k\geqslant 0}y_k\alpha^k$ in powers of the parameter $\alpha$. 

Then, substituting the following expansions
\begin{equation}
\frac{-y+3i\alpha y^2}{(1-6 i \alpha y)^{1/2}}=-y+\frac{9 \alpha^2 y^3}{2}+o\left(\alpha^2 \right)
,\,\,\,\,\,\,
\frac{\alpha}{(1-6 i \alpha y)^{3/2}}=\alpha +9 i \alpha^2y+ o\left(\alpha^2 \right),
\end{equation}
into equation (\ref{eqfory}) and using the ansatz for $y(\alpha)$, we obtaining the relations for the first three coefficients
\begin{equation*}
y_0=0,\,\,\,\,\,\,
y_1=-b_1=-3i,\,\,\,\,\,\,
y_2=\frac{1}{2} \left( 9 y_0^3-18 i b_1 y_0 \right)=0.
\end{equation*}

Hence, we have $y=-3i \alpha+o\left(\alpha^2 \right)$, and the integral in the first approximation
\begin{equation*}
	\begin{split}
		J\left( \alpha \right)=S[y]-\underbrace{\frac{1}{2}\mathrm{ln}(1-6 i \alpha y) }_{\frac{1}{2}\left(-6i\alpha y +o\left( \alpha^2 \right) \right)}+\underbrace{W\bigg(\frac{-y+3i\alpha y^2}{\sqrt{1-6 i \alpha y}},\frac{\alpha}{\sqrt{1-6 i \alpha y}^3}\bigg)}_{c_0+c_1\alpha+ c_2 \alpha^2+o \left( \alpha^2 \right)}=-\frac{15}{2} \alpha^2+o \left(\alpha^2 \right).
	\end{split}
\end{equation*}
gives the same result, obtained in Exercise \ref{d3}, where Lemmas \ref{lem6} and \ref{lem8} have been used.

\section{Conclusion}\label{conc}
Now we would like to give some comments on the transition from the simplest cubic model to an arbitrary one. First of all we note that the diagram equalities from Sections \ref{diag} and \ref{gen}, the Dyson--Schwinger equation, the Legendre transformation have a general nature. They are valid for all models. So to transit from one model to other we need to change only several basic blocks, that are presented in the  following table.
\renewcommand{\arraystretch}{1.5}
\begin{center}
\begin{tabular}{ | c | c | c | }
	\hline
	 & Cubic "toy" model & Arbitrary model \\ \hline
	Variable set & $k\in\mathbb{N}$ & $z\in\mathcal{M}$, $\mathcal{M}$ is a manifold \\
	\hline
	Field& $x_{(\cdot)}:\mathbb{N}\to\mathbb{R}$ & $\phi(\cdot)\in\mathcal{A}$, $\mathcal{A}$ is a functional space on $\mathcal{M}$\\
	\hline
	Classical action & $S[x]=\sum_{k\geqslant 1}(-x_k^2/2+i\alpha_k^{\phantom{3}} x_k^3)$& $S[\phi]$ is a functional on $\mathcal{A}$\\
	\hline
	Propagator & $\delta_{ij}$, $i,j\in\mathbb{N}$ & $G(z_1,z_2)$, $z_1,z_2\in\mathcal{M}$\\
	\hline
	Derivative & $\frac{\partial}{\partial x_k}$ & $\frac{\delta}{\delta\phi(z)}$ is a functional derivative\\
	\hline
	Path integral & $\lim\limits_{N\to+\infty}\int_{\mathbb{R}^n}\exp(S[\mathbb{X}_N])\,d\mathbb{X}_N$&
	$\int_{\mathcal{A}}\exp(S[\phi])\,\mathcal{D}\phi$, $\mathcal{D}\phi$ is a measure on $\mathcal{A}$\\
	\hline
	Vertices & $\includegraphics[scale=0.6]{ex}\sim i\alpha_k^{\phantom{3}}$ & $\zeta$ is a set of vertices \\
	\hline
\end{tabular}
\end{center}
\renewcommand{\arraystretch}{1}

Let us note that all calculations of Section \ref{res} are based on the model properties. For example, on the relation between \includegraphics[scale=0.6]{ex} and $i\alpha$. Hence, the calculations do not have general nature. However, the result obtained in Section \ref{res}, that $\Gamma_0$ and $\Gamma_1$ are the sums of strongly connected diagrams without and with one external line correspondingly, is the universal one.

Section \ref{back} is devoted to the background field method. Of course, we again did all calculations by using the simplest model. But the main concept, that the method consists in the shift of variable by a special type function, is preserved for all the examples.

It is also worth noting that these lecture notes contain only basic and elementary equalities. An analysis of more complex models and their properties can be found in the list of references given above.

\paragraph{Acknowledgements.}
This research is supported by the grant in the form of subsidies from the Federal budget for state support of creation and development world-class research centers, including international mathematical centers and world-class research centers that perform research and development on the priorities of scientific and technological development. The agreement is between MES and PDMI RAS from \textquotedblleft8\textquotedblright\,November 2019 \textnumero\ 075-15-2019-1620. Also, A.V. Ivanov is a winner of the Young Russian Mathematician award and would like to thank its sponsors and jury.

\end{document}